%% file: main.tex
\documentclass[runningheads]{llncs}

\usepackage{cite}
\usepackage{graphicx}
\usepackage{array}
\usepackage{xspace}

\usepackage{amsmath}
\usepackage{amssymb}
\usepackage{enumitem}
\setlist[itemize]{topsep=0.1em, parsep=0em, itemsep=0.1em}
\setlist[enumerate]{topsep=0.1em, parsep=0em, itemsep=0.1em}
\usepackage[frozencache=true,cachedir=minted-cache]{minted}
\setminted{autogobble, breaklines, numbersep=6pt, xleftmargin=0pt, fontsize=\fontsize88} 
\usepackage{cite}
\usepackage{booktabs}
\usepackage{pgfplots}
\usepackage{pgfplotstable}
\pgfplotsset{compat=1.18}
\usepackage{siunitx}
\usepackage{subcaption}

\usepackage{inconsolata}

\usepackage[colorlinks, citecolor=purple]{hyperref}

\usepackage[capitalise]{cleveref} 

\crefname{equation}{}{}

\allowdisplaybreaks[1]

\newtheorem{notation}{Notation}[theorem]
\newtheorem{myrule}{Rule}  
\crefname{myrule}{Rule}{Rules}

\NewDocumentCommand {\E} { o } {%
    \mathbb{E}{\IfNoValueTF{#1} {} {\left[\,#1\,\right]}}
}
\RenewDocumentCommand {\Pr} { o } {%
    \mathrm{Pr}{\IfNoValueTF{#1} {} {\left[\,#1\,\right]}}
}

\let\+=\mathbf

\let\t\texttt
\newcommand\clang{\t{clang}\xspace}
\newcommand\oclang{\textsf{obliv-clang}\xspace}
\newcommand\olevel{\t{olevel}\xspace}
\def\olevel{\t{olevel}\xspace}

\newif\ifcr

\begin{document}

\title{obliv-clang: Real-World Oblivious Programming in C++}

\ifcr
\else
\subtitle{(Extended Version)}
\fi

\author{Yunqian Luo\inst{1}\orcidID{0009-0000-3366-3722}
\and
Mingyu Gao\inst{1,2}\orcidID{0000-0001-8433-7281}
}

\institute{Tsinghua University, Beijing, China  \\
\and
Shanghai Qi Zhi Institute, Shanghai, China \\
\vspace{4pt}
\email{luoyq23@mails.tsinghua.edu.cn, gaomy@tsinghua.edu.cn}}

\maketitle

\begin{abstract}
  \input{chapters/00_abstract}

\end{abstract}

\input{chapters/10-intro}

\input{chapters/20-background}

\input{chapters/21-semantics}
\input{chapters/30-soundness}
\input{chapters/40-implementation}
\input{chapters/41-evaluation}

\input{chapters/50-related}
\input{chapters/51-conclusion}

\begin{credits}
\subsubsection{\ackname}
The authors thank the anonymous reviewers for their valuable suggestions, and the Tsinghua IDEAL group members for constructive discussion.
The technical solution and experimental design presented in this paper were independently completed by the authors. AI tools were used solely for language polishing and formatting optimization, and did not participate in the development of the research ideas or core content.
\end{credits}

\bibliographystyle{splncs04}
\bibliography{ideal}

\end{document}

%% file: chapters/00_abstract.tex
Side-channel vulnerabilities, particularly timing and access-pattern-based attacks, have become critical issues for confidential data processing in trusted environments.
Oblivious programming is an effective approach to alleviate these attacks by making program execution not leak any secret through execution time and data access traces.
To facilitate oblivious programming in practice, we propose a compilation-time checking tool, \oclang, which can comprehensively check the obliviousness of a program written in C++. 
It is designed to support the rich language features in C++, including the complicated concept of arbitrarily nested pointers, in order to seamlessly work with existing industry-level codebases and produce high-performance compiled binaries with minimum compilation overheads.
We design a set of rules in \oclang and formally prove their soundness in the presence of complicated C++ language features.
We also implement several non-trivial oblivious algorithms as case studies to demonstrate the expressiveness of \oclang, and show that programs compiled using \oclang can outperform previous solutions.

%% file: chapters/10-intro.tex
\section{Introduction}\label{sec:intro}

\ifcr
\else
The rapid development of artificial intelligence and big data analytics applications has substantially promoted the profit of high-quality data. Consequently, it has also raised the issue of data privacy protection. The typical way to handle confidential data without leaking sensitive secrets is to process them in a trusted environment. An organization could maintain its own computing servers and control all the hardware and system software, e.g., the operating system (OS). Alternatively, with cloud computing and outsourced computing, they can leverage hardware processors equipped with trusted execution environments (TEEs), which can isolate sensitive execution from untrusted environments. TEEs allow computing on even untrusted machines to become trusted.

However, even in a trusted environment, without careful programming, an attacker can still obtain partial or complete information about the secrets by performing various side-channel attacks. Researchers have published a wide range of such attacks, utilizing shared caches~\cite{gotzfriedCacheAttacksIntel2017}, OS-controlled page tables~\cite{203868,shihTSGXEradicatingControlledChannel2017a}, hardware branch predictors~\cite{leeInferringFinegrainedControl2017}, speculative execution~\cite{vanbulckForeshadowExtractingKeys2018}, etc.
Such attacks are possible because while the sensitive data are isolated at the architecture level, certain properties of the execution could still be leaked through timing observation and shared hardware units. If these side-channel properties depend on the secret values, the attacker can infer valid information from external observation without breaking the architecture isolation.

Various approaches have been proposed to alleviate side-channel attacks, either with extra hardware support~\cite{deutschDAGguiseMitigatingMemory2022, liuGhostRiderHardwaresoftwareSystem2015, kimSTEALTHMEMSystemLevelProtection2012, raneRaccoonClosingDigital2015, yuDataObliviousISA2019a}, or using software techniques~\cite{shihTSGXEradicatingControlledChannel2017a, liuMemoryTraceOblivious2013, cauligiFaCTDSLTimingsensitive2019b, almeidaJasminHighAssuranceHighSpeed2017, sinhaCompilerVerifierPage2017, moodFrigateValidatedExtensible2016, liuObliVMProgrammingFramework2015a}.
Hardware methods may provide more fundamental solutions and less performance loss. But applying hardware changes is quite expensive and usually takes long time. In addition, the required fixes may sometimes conflict with microarchitecture-level performance optimizations and sacrifice the speed of non-sensitive execution.
On the other hand, software solutions focus on making the side-channel information (such as execution time and I/O signals) of a trusted program execution not reveal any secret. This property of a program is called ``obliviousness,'' and oblivious programming aims to write the code in a way to achieve such obliviousness.
Today, hardware and software defenses can be used together to target different side channels; timing and I/O pattern-related issues are best avoided by software, while hardware vulnerabilities like speculation execution require hardware modifications.
\fi

\ifcr
\else
This work aims to design better toolchain support for oblivious programming.
Existing methods for implementing oblivious programs can be classified into several categories, each with its own drawbacks.
Some solutions~\cite{zahurOblivCLanguageExtensible2015, liuObliVMProgrammingFramework2015a, cauligiFaCTDSLTimingsensitive2019b, liuAutomatingEfficientRAMModel2014} propose domain-specific languages for obliviousness, but integrating a new language and its toolchain into existing codebases and development workflow is expensive, if not impractical, especially in an industrial environment with complicated software projects.
Others~\cite{sonObliCheckEfficientVerification2021,zhangCTLLVMAutomaticLargeScale} rely on separate modeling tools to check program obliviousness without impacting the original compilation flow. But ensuring the consistency between the descriptive model and the actual implementation of a program is challenging and highly error-prone.
The last direction~\cite{sinhaCompilerVerifierPage2017, raneRaccoonClosingDigital2015} is to modify the compilers and apply automatic program transformations during compilation. However, without knowing the high-level application semantics, code transformations can only be done conservatively and result in significant performance degradation, or simply cannot be done at all in complex cases.
\fi

\ifcr
\else
In this work, 
\fi
we present \oclang, a compiler plugin for the widely used, industry-class C++ language, which checks the obliviousness of a program statically at the compilation time. 
It is implemented using the popular \clang and LLVM infrastructure. 
As the key principle, \oclang is designed to work directly with C++ and support its rich language features, including pointers and references, function calls, compound types, object-oriented programming, templates, and more. Therefore, it can seamlessly work with existing codebases with minimum porting effort, eliminating the extra burden of using a new domain-specific language in previous work. 
In addition, \oclang is implemented directly in the compiler, so it avoids the consistency issue between separate modeling and implementation. 
Finally, using a familiar, mature, and industry-class language like C++ allows developers to conveniently write their code, and also able to retain high performance for the developed oblivious programs. 

The key challenge of designing a compilation-time checker like \oclang is its soundness, meaning that the check must be comprehensive so that if a program passes, it should be considered safe to execute. The major difficulty is to handle the arbitrarily nested pointers (like ``a pointer to a pointer to a pointer'') in C++, which can be dereferenced, taken-address, and assigned in many ways. In particular, we find that assigning a secret pointer with a public pointer value, which seems benign, could sometimes be unsafe and leak sensitive data (see \cref{example:dangerous-assignment} in \cref{sec:semantics:olevel}).
To address this problem, \oclang formally introduces the concept of oblivious levels (\olevel{}s), together with a set of rules to deduce each variable's \olevel along the program dataflow and to forbid unsafe operations accordingly. These rules are extended beyond simple operations to support also function calls, compound types, and templates.
In such a way, \oclang is able to provide verifiable security with formal proofs. 

Since \oclang is a front-end checker, we also need to ensure that the backend compiler optimization passes do not break the obliviousness property. We thoroughly analyze the relevant optimization passes in LLVM and suppress the ones that may compromise obliviousness. Furthermore, we apply an extra double-check on the generated assembly code using existing tools~\cite{disselkoenFindingEliminatingTiming}, which ensure end-to-end obliviousness.

We have implemented several non-trivial oblivious algorithms, including PathORAM~\cite{stefanov2018path} and Oblivious BFS~\cite{blantonDataobliviousGraphAlgorithms2013}, using our \oclang tool, in order to demonstrate its good expressiveness. 
We also quantitatively evaluate \oclang, showing that the extra obliviousness check incurs minor overheads (less than 18\%) to the compilation time, and it produces relatively high-performance compiled binaries. When compared to previous tools, the programs compiled with \oclang run 27.8\% faster than Jasmin~\cite{almeidaJasminHighAssuranceHighSpeed2017} and 66.9\% faster than FaCT~\cite{cauligiFaCTDSLTimingsensitive2019b} on average. The implementation of \oclang has been open sourced at \href{https://github.com/tsinghua-ideal/obliv-clang}{github.com/tsinghua-ideal/obliv-clang}.

Our contributions in this paper are summarized below:
\begin{itemize}
    \item The design and implementation of a program obliviousness check tool, \oclang, as a compiler plugin working seamlessly with the C++ language.
    \item A set of obliviousness check rules that support the rich set of C++ language features, including arbitrarily nested pointers, function calls, compound types, templates, etc.
    \item A formal proof about the soundness of our obliviousness check rules.
    \item The case studies and evaluation experiments to demonstrate the expressiveness and performance advantages of \oclang.
\end{itemize}

\ifcr\else
The rest of this paper is organized as follows.
\cref{sec:back} provides background and discusses prior tools' limitations.
\cref{sec:semantics,sec:impl} present our design and describe implementation details.
\cref{sec:soundness} gives the soundness proof.
\cref{sec:suppress} discusses how to suppress backend optimization passes to avoid break obliviousness.
\cref{sec:eval} illustrates several case studies and evaluates our design.
\cref{sec:related} discusses related work and \cref{sec:conclusion} concludes the paper.
\fi

%% file: chapters/20-background.tex
\section{Background and Motivation}
\label{sec:back}

\ifcr\else
In this section, we first introduce the concept of oblivious programming and how it mitigates timing- and access-pattern-based information leakage in trusted hardware (\cref{sec:back:back}), and then discuss the limitations of existing approaches to motivate our work (\cref{sec:back:limits}). We next describe our high-level approach and summarize the key challenges (\cref{sec:back:approach}). Lastly, we specify our security model by formally defining properties that oblivious programs must satisfy (\cref{sec:back:model}).
\fi

\subsection{Oblivious Programming}
\label{sec:back:back}

Program obliviousness prevents attackers from learning sensitive information through execution traces (timing, memory, and I/O patterns). Attackers extract traces either directly (through privileged access) or indirectly via shared hardware~\cite{gotzfriedCacheAttacksIntel2017, 203868, shihTSGXEradicatingControlledChannel2017a}.
Oblivious programming requires branch conditions, memory, and I/O patterns to remain independent of sensitive data. Common techniques include redundant branch execution, oblivious memory scanning, and \t{cmove}-like instructions. Researchers have developed efficient algorithms including ORAM~\cite{stefanov2018path}, oblivious BFS~\cite{blantonDataobliviousGraphAlgorithms2013}, and oblivious data structures~\cite{wang2014oblivious}.

Oblivious programming is tricky and subtle---developers can easily break obliviousness through overlooked details. Thus, automatic tools are essential to facilitate correct oblivious programming.

\subsection{Tradeoffs and Limitations of Existing Tools}
\label{sec:back:limits}

\ifcr
\else
Previous work on compilation-time obliviousness guarantees can be classified along two dimensions. The first considers whether program behaviors are modified: some tools only perform checks~\cite{almeidaJasminHighAssuranceHighSpeed2017, sonObliCheckEfficientVerification2021, reparazDudeMyCode2016}, while others actively transform programs to hide non-oblivious branches and memory accesses~\cite{liuObliVMProgrammingFramework2015a, liuMemoryTraceOblivious2013, cauligiFaCTDSLTimingsensitive2019b, raneRaccoonClosingDigital2015, sinhaCompilerVerifierPage2017, shindePreventingPageFaults2016}.
Automatic transformations simplify programming but introduce complexity, as generating dummy versions for program segments is not always straightforward. In many cases the compiler either has to apply general, conservative transformations that result in unnecessary performance overheads, or simply cannot handle irregular patterns.

The second dimension involves the level at which checks or transformations occur. Some tools target abstract syntax tree (AST) structures in the compiler frontend~\cite{cauligiFaCTDSLTimingsensitive2019b, sonObliCheckEfficientVerification2021, liuMemoryTraceOblivious2013,sinhaCompilerVerifierPage2017}, while others work on low-level intermediate representations (IRs), particularly LLVM IR, in the backend~\cite{raneRaccoonClosingDigital2015,shindePreventingPageFaults2016}.
Frontend approaches preserve more source code information, facilitating high-level analysis to achieve better efficiency. But in these tools, the unchanged backend optimizers may later alter the supposed-to-be oblivious behaviors and thus undermine security.

Besides the above tradeoffs, existing approaches for oblivious programming also have other key limitations.
Many existing designs propose domain-specific languages (DSLs) to simplify compilation-time transformations. However, DSLs require costly toolchain development and lack standard features (editor support, build management, debugging tools, extensive libraries). Additionally, most DSL-based tools focus on small computation kernels, lacking expressiveness for complex systems.
\fi

\subsection{Our Approach}
\label{sec:back:approach}

We build upon C++ and the \clang compiler, verifying obliviousness directly on source code at compilation time. C++ is widely used in industry. It offers well-defined semantics in its specification~\cite{StandardDraftSources2022}, a minimal runtime with close hardware control, and compatibility with numerous software interfaces and libraries.
The \clang compiler offers a friendly, extensible plugin interface used by many tools~\cite{WhatClangd2023, ClangTidyExtraClang2023}, enabling non-intrusive checker implementation.
The \clang frontend processes C++ and C code in a unified AST format, allowing \oclang to handle C code seamlessly---important since many system and cryptographic libraries are written in C.
While our tool is a checker, we also implement backend mechanisms to guarantee obliviousness during code generation (\cref{sec:suppress}).

\ifcr
\else
Rust might also seem attractive given its memory safety guarantees. Nevertheless, we chose C++ for two reasons. First, Rust's behavior lacks comprehensive specification---C++ has a detailed specification describing program semantics, while Rust relies on the latest \t{rustc} implementation as its standard. Second, \t{rustc} currently lacks a stable plugin interface, having deprecated its previous interfaces in 2018~\cite{aturonTrackingIssuePlugin2015}.

\textbf{Challenges.}
Working directly with C++ presents several challenges due to the necessity of comprehensively handling its rich and flexible semantics.
First, a primary difficulty involves handling pointers and references, where expressing secrecy requirements for objects with arbitrarily nested pointers requires a generic, precise, and concise scheme.
Second, function calls present another challenge, as implementations and call sites typically reside in different translation units compiled separately, necessitating obliviousness-aware interfaces to ensure consistency.
Third, supporting object-oriented programming requires consideration of relationships between compound class objects, their members, and methods.
Fourth, templates could generate multiple declarations from a common template, in which some may work on sensitive data and have oblivious requirements, while others do not.
Finally, some cases require external knowledge that cannot be analyzed solely from code, such as calls to external libraries.
Our design in \cref{sec:semantics} addresses these challenges.
\fi

\subsection{Security Model}
\label{sec:back:model}

Our security model focuses on timing and access-pattern-related attacks in two scenarios:
\begin{enumerate}
  \item The victim program executes in a TEE enclave, where adversaries with hardware access can perform access-pattern-based attacks, including timing, cache~\cite{gotzfriedCacheAttacksIntel2017}, and page-table~\cite{203868,shihTSGXEradicatingControlledChannel2017a} attacks.
  \item The victim program executes on normal hardware, where adversaries without hardware access can interact with the program and observe timing.
\end{enumerate}

Side-channel attacks based on power~\cite{lippPLATYPUSSoftwarebasedPower2021}, electromagnetism~\cite{sayakkaraSurveyElectromagneticSidechannel2019}, speculative execution~\cite{vanbulckForeshadowExtractingKeys2018}, rowhammer~\cite{jangSGXBombLockingProcessor2017}, and similar vectors are beyond our scope. 
\ifcr\else
Preventing these vulnerabilities solely through oblivious programming is generally considered impossible or prohibitively expensive. Corresponding countermeasures are studied in orthogonal work~\cite{sonMakingDRAMStronger2017, zhangPowerSideChannels2015}.
\fi

To ensure security under these assumptions, programs must satisfy the following properties in \cref{thm:sec,thm:dataflow-informal}. 
\ifcr\else
The first property ensures adversaries cannot learn secret information from execution traces, while the second prevents secret leakage to the public domain. 
\fi
Formal definitions of object secrecy will be provided in \cref{sec:semantics}.

\begin{theorem}[Secrecy]\label{thm:sec}
  For a program $P$ with no undefined behaviors, and any external input $x$,
  let $v_\text{pub}(x)$ be the set of values of its public objects during execution, $\textsf{Trace}(P, x)$ be the execution trace.
  There exists a probabilistic polynomial time (p.p.t.) simulator $\textsf{Simu}$ such that no p.p.t. adversary $\mathcal A$ can distinguish $\textsf{Trace}(P, x)$ and $\textsf{Simu}(v_\text{pub}(x))$.
  Formally:
  \begin{equation}
    \forall x, \Big\lvert
    \Pr[\mathcal A(v_\text{pub}(x), \textsf{Trace}(P, x))]
    -\Pr[\mathcal A(v_\text{pub}(x), \textsf{Simu}(v_\text{pub}))]
    \Big\rvert < \epsilon(n),
  \end{equation}
  where $n$ is a security parameter, and $\epsilon$ is a negligible function.
\end{theorem}

\begin{theorem}[Dataflow (informal)]\label{thm:dataflow-informal}
  There is no data flow from secret objects to public objects.
\end{theorem}

%% file: chapters/21-semantics.tex
\section{Design}
\label{sec:semantics}

We designed \oclang, a C++ compiler plugin for \clang that verifies program obliviousness directly on source code at compilation time.
As a checker, \oclang's primary goal is \emph{soundness}: programs passing its checks should be considered safe.
This section details \oclang's semantics, including the required programmer annotations to indicate protected data and the comprehensive rules enforced to ensure obliviousness. We focus on addressing challenges of supporting flexible C++ semantics---pointers, references, function calls, and compound types---while maintaining compatibility with external libraries.

\subsection{Getting Started: Oblivious Variables}

The basic information required by \oclang is the knowledge of which variables should be treated as secret. Since a compiler cannot inherently determine which objects are sensitive, we require programmer annotations in the source code to provide this information. We believe such annotations only add moderate burden to programmers, as programmers typically well understand the purposes and secrecy requirements of their variables.

Beginning with a minimal language subset (no references, no nested pointers, only plain values), an object's secrecy can be described by a boolean flag. \oclang requires secret variables to be annotated with the custom \t{obliv} attribute.
C++ specifications allow compilers to define custom attributes as language extensions~\cite[Chapter 9.12]{StandardDraftSources2022}, which can be applied to declarations, statements, translation units, and other elements.
Furthermore, with the \clang toolchain, parsing attributes can be easily done using plugins, without intrusive changes to the compiler's internal core logic (\cref{sec:impl}).

To satisfy the Secrecy and Dataflow properties from \cref{sec:back:model}, the following rules are required:

\begin{myrule}[Basic Arithmetic]\label{rule:basic-arith}\normalshape Any expression depending on an oblivious object is also oblivious. Note that we only consider value dependencies; comma expressions, compilation-time constant expressions (e.g., \t{sizeof}, \t{alignof}), and take-address operations are not considered dependencies.
\end{myrule}

\begin{myrule}[Basic Casting]\label{rule:basic-casting} \normalshape It is forbidden to cast an oblivious expression to a non-oblivious expression.
\end{myrule}

\begin{myrule}[Basic Dereference]\label{rule:basic-deref}\normalshape It is forbidden to dereference an oblivious pointer. Similarly, it is forbidden to take an oblivious object as the base address or index in an array subscript expression.
\end{myrule}

\begin{myrule}[Allocation]\label{rule:allocation}\normalshape It is forbidden to create an object whose size is oblivious (e.g., \t{new[]} with oblivious length).
\end{myrule}

\begin{myrule}[Condition]\label{rule:condition}\normalshape The condition in any conditional clause (\t{if}, \t{while}, \t{for}, \t{switch}, ternary expressions, short-circuit boolean expressions) must not be an oblivious expression.
\end{myrule}

Note that ``cast'' in \cref{rule:basic-casting} refers to both explicit and implicit casting. In assignments, the right side is first cast into the target type before assignment to the left side, so this can also be considered as an assignment rule.

The following code snippet demonstrates these rules:

\begin{example}\label{example:obliv-var}
	~
	\begin{minted}{cpp}
    #define OBLIV [[obliv]]  // to make it look less cumbersome

    // declare an oblivious variable
    OBLIV int x = 1729;
    // error: cannot assign oblivious value to public variable
    int y = x + 3;
    // ok
    OBLIV int z = x * 10;
    // error: cannot take oblivious value as control-flow condition
    if (x > 0) x = x - 1;
    // error: cannot take oblivious value as array subscript
    int *arr = new int[16]; int w = arr[x % 16];
\end{minted}
\end{example}\nopagebreak

\subsection{Obliviousness Levels}
\label{sec:semantics:olevel}

For a pointer, obliviousness has dual meanings: the pointer's own value (which memory address it points to) and its dereferenced value (content at that address). With arbitrarily nested pointers in C++, this becomes increasingly complex.

We propose the concept of \emph{obliviousness level} (\olevel) to describe multi-level obliviousness. An expression's \olevel is a binary string where the $i$-th digit (counting from the right, starting from 0) represents the obliviousness after dereferencing it $i$ times. An expression whose \olevel ends with 1 is considered oblivious. In our discussion, an expression is less/more oblivious if its last \olevel digit is less/more than another's; an expression has no greater/smaller \olevel if every digit of its \olevel is no greater/smaller than the corresponding digit of another; ``greater/smaller'' is defined as the contrary of ``no greater/smaller.''

Dataflow occurs with assignment expressions. To satisfy the dataflow property, we should not allow assigning an expression with another of a greater \olevel. However, assignments with smaller \olevel can also be dangerous. In \cref{example:dangerous-assignment}, \t{ppc} of \olevel \t{100} is assigned with \t{\&public\_c} of smaller \olevel \t{0}:

While assignments between different \olevel{}s appear unsafe, some cases with smaller \olevel{}s seem benign. For example, assigning a public value to a secret value, is clearly safe --- we can not pass any secret information into public by this assignment. But when pointer nests get deeper, it is much less oblivious whether an assignment is safe.

\begin{example}\label{example:dangerous-assignment}
	~
	\begin{minted}{cpp}
  char public_c = 'p';
  [[obliv("1")]] char secret_c = 'x';
  [[obliv("10")]] char *pc = &secret_c;
  [[obliv("100")]] char **ppc = &pc;

  // if obliv("00") to obliv("10") is allowed
  *ppc = &public_c;
  *pc = secret_c;  // 'x' flows to public_c!
  \end{minted}
\end{example}\nopagebreak

\ifcr
\else
To determine permissible assignments, we must rigorously define our dataflow theorem. Undesired dataflow occurs when oblivious values flow into non-oblivious memory cells. Since we allow casting/assignments between different \olevel{}s, we must distinguish between the \olevel embedded in the expression type and the actual \olevel carried at runtime.

In C++, expressions are classified as lvalues (can be taken address and can appear on an assignment left side if not constant) or rvalues (cannot be taken address or appear on an assignment left side). Since every lvalue corresponds to a memory area, its ``actual'' \olevel{}---the \emph{inherent \olevel}---is defined as the \olevel of its corresponding memory allocation. For stack/static allocations, this is the \olevel of the corresponding variable; for heap allocations, it is the \olevel of the corresponding \t{new} expression.

We call lvalue expressions corresponding to memory allocations (local/global variables, dereferences of \t{new} expressions) \emph{primitive lvalues}. Two lvalues can \emph{alias} each other, meaning they correspond to the same memory area. In a well-formed program, every expression corresponds to allocated memory, thus aliasing with a primitive lvalue.

For rvalues, obliviousness depends on the operator and operands according to rules we will elaborate below. We do not define inherent \olevel for rvalues, but for a pointer rvalue \t{p}, we can define the inherent \olevel of \t{*p} since \t{*p} is a lvalue.

With these concepts, we formally describe \cref{thm:dataflow-informal} as two theorems:
\fi

\begin{theorem}[Dataflow of Assignment]\label{thm:dataflow}
	There is no execution of an assignment statement whose assignee (which must be an lvalue) is aliasing a non-oblivious lvalue, but the assigner is an oblivious value.
\end{theorem}

\begin{theorem}[Dataflow of Value Dependency]\label{thm:dataflow-dep}
	An expression whose value depends on oblivious values is also oblivious.
\end{theorem}

\newcommand{\alias}[1]{\stackrel{#1}{\sim}}
\newcommand{\assigned}[1]{\stackrel{#1}{\gets}}
\newcommand{\down}[1]{\downarrow^{#1}}
\newcommand{\up}[1]{\uparrow^{#1}}
\ifcr
\else
\begin{notation}
	We use the notations below for convenience:
	\begin{itemize}
		\item For an lvalue $x$ and a non-negative integer $k$, denote by $x \down{k}$ the expression obtained by indirecting $x$ by $k$ times. In other words, $x \down{0}$ is $x$, and $x \down{k+1}$ is \t{*($x\down{k}$)}. Similarly, we denote by $x \up{k}$ the expression obtained by taking-address (i.e., the unary \texttt{\&} operator) on $x$ by $k$ times.
		\item For two lvalues $x$, $y$ and time $t$, if at time $t$, the expressions $x$ and $y$ correspond to the same memory cell, we say that $x$ and $y$ are aliased at $t$, denoted by $x \alias{t} y$. The superscript $t$ may be omitted when the context is clear.
		\item For an lvalue $x$, an expression $e$, and time $t$, denote by $x \assigned{t} e$ that $x$ is assigned with $e$ at time $t$ (may come with an implicit casting). Similarly, $t$ may be omitted.
		\item For an expression $x$, denote by $\t{const}(x)$ and $\t{obliv}(x)$ its obliviousness and constantness, whose value is either 0 or 1.
	\end{itemize}
\end{notation}
\fi

Now we specify the \olevel rules to satisfy \cref{thm:dataflow}. The danger in \cref{example:dangerous-assignment} occurs because an oblivious expression \t{**ppc} aliases with a non-oblivious expression \t{public\_c}, but \t{*ppc} is not constant and can be assigned a secret value which flows to the non-oblivious \t{public\_c}. To avoid this, we restrict ``non-oblivious to oblivious'' assignments by ensuring the ancestors in the pointer hierarchy are constant:

\begin{myrule}[Casting]\label{rule:casting}\normalshape
	Consider two expressions, $a$ and $b$. $b$ can be implicitly cast to $a$ if and only if, in addition to C++'s cast rules, the following \emph{oblivious cast rules} hold: \begin{enumerate}
		\item For a $k > 0$, if $\t{const}(a \down{k}) = 0$, for any $k' \geq k$, \begin{align}
			      \t{const}(a \down{k'}) & = \t{const}(b \down{k'}) \label{eq:equal-const}  \\
			      \t{obliv}(a \down{k'}) & = \t{obliv}(b \down{k'}). \label{eq:equal-obliv}
		      \end{align}
		\item For any $k \geq 0$, \begin{align} \label{eq:rule:casting:2}
			      \t{obliv}(a \down{k}) \geq \t{obliv}(b \down{k}).
		      \end{align}
	\end{enumerate}

	For an assignment $a \gets b$, if $a$ is a reference type, the assignment is allowed when $b$ can be implicitly converted to the \olevel and type of $a$ according to the above casting rules. If $a$ is a value type, $b$ only needs to be converted to the \olevel and type of constant $a$, i.e., with $\t{const}(a)$ set to $1$ but other \olevel and type unchanged.
\end{myrule}

Intuitively, condition 2 in \cref{rule:casting} prevents assigning oblivious values to non-oblivious targets, e.g., assigning an oblivious $b$ to a non-oblivious $a$.
Condition 1 parallels C++'s implicit conversion rule for \t{const}-modifiers~\cite[Chapter 7.3.6]{StandardDraftSources2022}, preventing inherent-constant memory from being overwritten. When $\t{const}(a\down k) = 0$ and $\t{obliv}(a \down {k'}) \neq \t{obliv}(b \down{k'})$ for some ${k'} \geq k$, assigning an oblivious value to $a \down{k'}$ would also assign it to the non-oblivious $b \down {k'}$, violating the dataflow theorem. 
\ifcr\else
We prove the soundness of \cref{rule:casting} in \cref{sec:soundness}.
\fi

To satisfy \cref{thm:dataflow-dep}, we have:

\begin{myrule}[Arithmetic]\label{rule:arith}\normalshape For an arithmetic operator expression, if any operands are oblivious, the total expression is also oblivious. Note that we only consider arithmetic expressions; comma expressions, compilation-time constant expressions (e.g., \t{sizeof}, \t{alignof}), and pointer reference/dereference operators are not considered dependencies.
\end{myrule}

\begin{myrule}[Dereference]\label{rule:deref}\normalshape It is forbidden to dereference an oblivious pointer (the array indexing expression \t{p[i]} is interpreted as \t{*(p + i)}). The result of a dereference operator shifts the \olevel right by one digit, while the result of a take-address operator shifts the \olevel left by one digit.
\end{myrule}

\begin{myrule}[Pointer Arithmetic]\label{rule:pointer-arith}\normalshape
	For the arithmetic of pointers, the only allowed type is the addition of a pointer \t{p} and an integer \t{i} (including the equivalent \t{\&p[i]}). The 0th digit of the result \olevel is the logical-OR of the operands' obliviousnesses, and the remaining digits are the same as \t{p}.
\end{myrule}

These rules ensure the theorems stated earlier. Note that \cref{rule:casting,rule:arith,rule:deref} replace their ``basic'' versions (\cref{rule:basic-casting,rule:basic-arith,rule:basic-deref}). The following example demonstrates how oblivious assignment rules work:

\begin{example}\label{example:assignment-rule}
	~
	\begin{minted}{cpp}
    // [[obliv]] is the abbreviation of [[obliv("1")]]
    [[obliv]] int x = 1;
    [[obliv("10")]] int *xp = &x;
    [[obliv("100")]] int **xpp = &xp;
    // error, 2nd digit of olevel differs
    [[obliv("110")]] int **ypp = xpp;
    // ok, 2nd level is constant
    [[obliv("110")]] int *const *const ypp = xpp;  
    // ok, the first level is passed by value
    [[obliv("101")]] int **ypp = xpp;        
\end{minted}
\end{example}\nopagebreak

\subsection{Function Calls}
\label{sec:semantics:ofunc}

Since callers typically only see function declarations in headers while their implications could exist in other translation units, obliviousness must be expressed in function signatures. We interpret function calls as series of assignments:

\begin{myrule}[Function Call]\label{rule:function} \normalshape
	When checking a function, the following rules should be enforced:
	\begin{enumerate}
		\item Each parameter of a function can be annotated with an \olevel.
		\item When checking the function body, the parameters are treated as variables with their annotated \olevel{}s.
		\item The function itself can be annotated with an \olevel, which indicates the \olevel of its return value.
		\item When checking the function body, every expression returned in a \t{return} statement should be able to be cast to the annotated \olevel of the function, according to \cref{rule:casting}.
		\item When calling a function, for each parameter, the given argument should be able to be cast to the \olevel of the declared parameter, according to \cref{rule:casting}.
	\end{enumerate}
\end{myrule}

%
%
%

\subsection{Classes}
\label{sec:semantics:complex-type}

C++ classes (\t{struct}, \t{union}) declare compound types with methods, inheritance, and polymorphism, requiring additional rules:

\begin{myrule}[Class]\label{rule:compound}\normalshape
	For a class type (\t{struct} and \t{union} types also included), the following rules are enforced:
	\begin{enumerate}
		\item In the declaration of a class type, its members can be annotated with an \olevel. The difference from a normal \olevel declaration is that the last digit of class member \olevel can be \t{"x"} besides \t{0} and \t{1}. If not specified, a class member has an \t{"x"} \olevel.
		\item In a class member access expression (e.g., \t{c.member} or \t{c->member}), if the \olevel of the member declaration ends with \t{"x"}, the \olevel of the expression is the \olevel of the member declaration with the last digit changed to the \olevel of the class instance (\t{c} or \t{*c}). Otherwise the \olevel of the expression is the same as the member declaration.
		\item A method can be annotated with \t{obliv\_this}. In the presence of this annotation, \t{this} is interpreted as a pointer to an oblivious instance. An oblivious instance can only call \t{obliv\_this} methods. A public instance can call all methods.
		\item When a method is overridden during inheritance, the \olevel annotations should be identical to the annotations in the original declaration.
		\item Virtual methods should not be called from a pointer to an oblivious polymorphic pointer.
	\end{enumerate}
\end{myrule}

By default, a member's obliviousness inherits from the instance (\olevel ending with \t{"x"}), analogous to C++'s \t{const}/\t{mutable} distinction. Virtual calls on oblivious polymorphic pointers are forbidden as the VPTR would be oblivious.

\subsection{Templates}
\label{sec:semantics:templates}

A single template could generate multiple instantiations. Since template parameters may carry \olevel information (via function return values or parameters), \oclang checks each instantiation rather than the template declaration itself. This is straightforward as \clang generates separate AST for each instantiation.


%
%

\subsection{Unsafe Code Blocks}
\label{sec:semantics:unsafe}

We provide the \t{obliv\_unsafe} attribute to mark code blocks as unsafe, bypassing obliviousness checks. This approach, inspired by Rust's \t{unsafe} keyword, is necessary in several cases:
\begin{itemize}
	\item Inspecting internal secret values when debugging;
	\item Implementing subroutines whose obliviousness requires sophisticated mathematics beyond automatic checking (e.g., extracting public values from oblivious variables through encryption/hashing with secret keys, or asserting on execution paths that should not trigger);
	\item Working with third-party libraries whose implementations are in separate translation units invisible to the compiler.
\end{itemize}

Like Rust's unsafe blocks, this feature indicates areas requiring special programmer attention and human auditing. We anticipate that in the future, more expressive semantics will reduce the need for unsafe blocks.

\subsection{Known Limitations}
\label{sec:semantics:discussions}

\ifcr
\else
When using external libraries, since the compiler can only see their header files, it cannot determine the obliviousness of external symbols unless we have the header file annotated with \olevel{}s. To use external libraries to process objects with \olevel{}s, the headers must be modified with appropriate annotations---an unavoidable effort for programmers. Fortunately, this typically only involves adding \olevel annotations into the headers, a moderate burden.

Another limitation of \oclang involves function overloading. C++ allows multiple functions with the same name but different argument types, which will be given distinct mangled names during compilation. Since the \olevel is not part of the C++ types, declaring overloaded functions with identical argument types but different \olevel{}s is prohibited unless the ABI encodes \olevel{}s in the mangled names. For external libraries supporting arguments with different \olevel{}s, multiple wrapper functions with different names may be required.

\fi

\section{Implementation}
\label{sec:impl}

\oclang is implemented in approximately 800 lines of C++ code. It registers an annotation handler to parse \t{obliv}-related annotations and a frontend plugin that executes after \clang parses the source code into an AST. The plugin scans the AST and raises compilation errors for code violating our rules.

\oclang works like other \clang plugins: it compiles into a dynamic library loaded during compilation, requiring only one additional compiler option to load the plugin. Our implementation targets a wide range of LLVM monorepo versions (tested on both LLVM 13 and LLVM 20). \clang's AST data structure remains relatively stable, so minimal maintenance is needed to support newer versions.

%% file: chapters/30-soundness.tex
\ifcr\else

\section{Soundness}
\label{sec:soundness}

In this section, we aim to prove the soundness of the \oclang design described in \cref{sec:semantics}. That is to prove, if a program passes the check of \oclang according to all the rules specified in \cref{sec:semantics} (except \cref{rule:basic-casting,rule:basic-arith,rule:basic-deref} that have been replaced), it should be considered safe.

\subsection{Execution Model}

Before discussing the soundness, it is necessary to accurately specify well-defined semantics for programs. For C++, there is an ISO specification~\cite{StandardDraftSources2022} curated by the C++ Standard Committee that stipulates the behaviors of C++ programs. In the specification, the available memory consists of one or more sequences of contiguous bytes. The behavior of a C++ program is described as \emph{creating, destroying, referring to, accessing, and manipulating objects}. An object occupies a region of memory (maybe empty) within its lifetime. An object may contain other objects, called \textit{subobjects}, in three forms: member subobjects, base class subobjects, and array subobjects. An object that is an array of \t{unsigned char} and \t{std::byte} can provide storage for other objects. The execution of a C++ program is interpreted as a sequence of expression evaluations. A list of rules are specified to define the relation between expressions and subexpressions and restrict the order of evaluations.

However, the specification only cares about ``observable'' effects, such as I/O and the final state, but leaves out the execution process that leads to the final state. Since we focus on analyzing the information leaked during the execution process, additional assumptions are required to specify these behaviors. Recall that the program execution is interpreted as a sequence of evaluations. Here, we classify the evaluation steps into the following categories:
\begin{itemize}
  \item Create an object, implicitly or explicitly, which is evaluated on each object declaration. Note that this only includes the action of allocating memory and does not include initialization.
  \item Free an object, which terminates the lifetime of an object. Similarly, the execution of the destructor is not included.
  \item Read an object, which evaluates a reference to an object.
  \item Write an object, which evaluates an assignment to an object.
  \item Evaluate an expression, where the expression can be an operator expression, a function-call expression, an implicit/explicit casting expression, etc. All sub-components of the expression should be already evaluated. Note that all syntax sugars are desugared. For example, templates and overloading are resolved, operator expressions are resolved to the corresponding function calls if they are not built-in, and constructors and destructors of temporary objects are also resolved to the corresponding function calls.
\end{itemize}

For our purpose, we stipulate that any evaluation steps come with an evaluation trace. Two sequences of evaluation steps would produce the same evaluation trace if they (1) perform the same kind of operation; and (2) lie in the same memory position if the evaluation is creating, freeing, reading, or writing an object.
Furthermore, we require that the evaluation order does not depend on the values of the operands. Recall that in C++, the evaluation order of operands is an unspecified behavior. For example, in an expression \t{f(x) + g(y)}, the compiler can arbitrarily swap the order that \t{f} and \t{g} are executed. We also grant the compiler the freedom to choose the order, but it should not depend on the operand values.

We only consider programs with well-defined behaviors, i.e. without undefined behaviors (UBs). The C++ standard gives rigorous definitions of UBs, including pointer reinterpretation, out-of-bound memory accesses, pointer use-after-free, etc. The compilers do not give any guarantee for a program with UBs, and neither does \oclang. Fortunately, the functionality to detect UBs is already included in most C++ compilers and sanitizers~\cite{clangUndefinedBehaviorSanitizerClang162022, llvmAvailableCheckers2022}, so we can rely on them to report errors when encountering UBs.

\subsection{Soundness Proof}

To prove the soundness of our tool, we need to prove that any program that passes the check rules of \oclang should satisfy \cref{thm:sec,thm:dataflow,thm:dataflow-dep}. It should be noticed that these theorems are independent. \cref{thm:sec} prevents the secrets from being leaked via the execution trace; \cref{thm:dataflow} prevents programmers from leaking the secrets ``unexpectedly''; \cref{thm:dataflow-dep} prevents leaking the secrets with computations.

\cref{thm:sec} clearly holds. According to our C++ modeling, given the program content, the execution trace is solely determined by the values in the conditional clauses and the pointer dereferences (including array accesses). At the same time, all the values in conditional clauses and array indexing are public values, according to \cref{rule:condition,rule:deref}. Therefore, the execution trace can be simulated solely with public objects.

\cref{thm:dataflow-dep} is guaranteed by the \olevel deduction rule of C++ operators, described in \cref{rule:arith}.

The non-trivial part is \cref{thm:dataflow}. It requires knowledge about in which cases two lvalues can aliasing with each other. We first have the following lemma.

\begin{lemma}\label{lem:invariant}
  For any expression $a$, denote by $K(a)$ the minimum positive integer $k$ such that $\t{const}(a \down{k}) = 0$ (let $K(a) = +\infty$ if such $k$ does not exist). If $a \alias{t} b$, and $b$ is a primitive expression, then at time $t$,
  \begin{align}
    \forall\ k' \geq K(a),\quad & \begin{cases}
                                    \t{const}(a \down{k'}) & = \t{const}(b \down{k'})  \\
                                    \t{obliv}(a \down{k'}) & = \t{obliv}(b \down{k'}),
                                  \end{cases} \\
    \forall\ k' \geq 0,\quad    & \t{obliv}(a \down{k'}) \geq \t{obliv}(b \down{k'})
  \end{align}

  Note that the above relations imply that for any $k' \geq 0$, $\t{const}(a \down {k'}) \geq \t{const}(b \down {k'})$
\end{lemma}
\begin{proof}
  Assume there exist some tuples $(a, b, t)$ such that $a \alias{t} b$, $b$ is primitive, but $K(a)$ does not satisfy the relations in the lemma. Among all such tuples, pick the tuple $(a, b, t)$ with the minimum $t$. Clearly $a \neq b$, otherwise for any $k' \geq 0$, $\t{const}(a \down{k'}) = \t{const}(b \down{k'})$ and also $\t{obliv}(a \down{k'}) = \t{obliv}(b \down{k'})$. The lemma is satisfied.

  Notice that the relations given by the lemma only depend on the types of $a$ and $b$, so these relations do not vary by time given $a$ and $b$, and they also do not hold at time $t' = t - 1$. But since $t$ is the minimum time that the lemma fails, the only possibility is that at time $t'$, the precondition does not hold, i.e., we have $a \not\alias{t'} b$, and an assignment transforms non-aliasing $a$ and $b$ into aliasing. For this to happen, there exist some expressions $c$ and $d$, such that at time $t'$ for some $k > 0$
  \begin{align}
    c \assigned{t'} d, \quad
    a \,\up{k} \alias{t'} c, \quad
    d \,\down{k} \alias{t'} b
  \end{align}

  Because $a \,\up{k} \alias{t'} c$, there exists some primitive $p$ such that
  \begin{align}
    a\,\up{k} \alias{t'} p \label{eq:alias-ak-p} \\
    c \alias{t'} p \label{eq:alias-c-p}
  \end{align}

  Since $t$ is the minimum time that the lemma fails for any tuple $(a, b, t)$, at time $t'$, $K(a \up k)$ and $K(c)$ satisfy the lemma.

  According to \cref{eq:alias-ak-p} we have
  \begin{align}
    \forall\ k' \geq K(a \up k), \; & \begin{cases}
                                        \t{const}(a \down{k' - k}) & = \t{const}(p \down{k'})  \\
                                        \t{obliv}(a \down{k' - k}) & = \t{obliv}(p \down{k'}),
                                      \end{cases} \label{eq:ap_equal}                  \\
    \forall\ k' \geq 0, \;          & \t{obliv}(a \down{k' - k}) \geq \t{obliv}(p \down{k'}) \label{eq:ap_ob}
  \end{align}

  We rewrite \cref{eq:ap_equal} into
  \begin{equation}
    \forall\ k' \geq K(a \up k) - k, \;
    \begin{cases}
      \t{const}(a \down{k'}) & = \t{const}(p \down{k'+k}) \\
      \t{obliv}(a \down{k'}) & = \t{obliv}(p \down{k'+k})
    \end{cases} \label{eq:ap_equal_2}
  \end{equation}

  Since $c$ is assigned with $d$, $c$ is not constant, i.e. $K(c) = 0$. According to \cref{eq:alias-c-p} we have
  \begin{align}
    \forall\ k' \geq 0, \; & \begin{cases}
                               \t{const}(c \down{k'}) & = \t{const}(p \down{k'}) \\
                               \t{obliv}(c \down{k'}) & = \t{obliv}(p \down{k'})
                             \end{cases} \label{eq:cp_equal}
  \end{align}

  Similarly, since $d \,\down{k} \alias{t'} b$ and $b$ is primitive,
  \begin{align}
    \forall\ k' \geq K(d \down{k}),\; & \begin{cases}
                                          \t{const}(d \down{k + k'}) = \t{const}(b \down{k'}) \\
                                          \t{obliv}(d \down{k + k'}) = \t{obliv}(b \down{k'}),
                                        \end{cases} \label{eq:bd_equal}                    \\
    \forall\ k' \geq 0, \;            & \t{obliv}(d \down{k + k'}) \geq \t{obliv}(b \down{k'}\label{eq:bd_ob})
  \end{align}

  Consider the implicit conversion during the assignment $c \assigned{t'} d$. According to the rule of assignment (\cref{rule:casting}), and notice that $\t{const}(c \down{K(c \down1) + 1}) = 0$ by definition, we have
  \begin{align}
    \forall\ k' \geq K(c \down1) \!+\! 1, & \begin{cases}
                                            \t{const}(c \down{k'}) = \t{const}(d \down{k'}) \\
                                            \t{obliv}(c \down{k'}) = \t{obliv}(d \down{k'}),
                                          \end{cases} \label{eq:cd_equal}                    \\
    \forall\ k' \geq 0, \;              & \t{obliv}(c \down{k'}) \geq \t{obliv}(d \down{k'}\label{eq:cd_ob})
  \end{align}

  We rewrite \cref{eq:cd_equal} into
  \begin{equation}
    \forall\ k' \geq K(c \down1) - (k - 1),\;
    \begin{cases}
      \t{const}(c \down{k'+k}) = \t{const}(d \down{k'+k}) \\
      \t{obliv}(c \down{k'+k}) = \t{obliv}(d \down{k'+k})
    \end{cases} \label{eq:cd_equal_2}
  \end{equation}

  Letting $k' = 1 + K(c \down 1)$ in \cref{eq:cd_equal}, we have \begin{align}
    \t{const}(d\down{1 + K(c \down 1)}) = \t{const}(c \down{1 + K(c \down 1)}) = 0.
  \end{align}
  Hence \begin{align}
    K(d \down 1) \leq K(c \down 1) \label{eq:kd1_leq_kc1}
  \end{align}
  Similarly letting $k' = K(c \down k)$ in \cref{eq:cd_equal_2}, we have \begin{align}
    K(c \down k) \geq K(d \down k). \label{eq:kck_geq_kdk}
  \end{align}

  We know that \begin{align} \label{eq:ka_geq}
    K(a) \geq K(a \up k) - k.
  \end{align}

  \newcommand{\explain}[1]{\text{#1}}
  For any $k' \geq 0$,
  \begin{align}
    \t{const}(a \down{k'})
     & \geq \t{const}(p \down{k + k'}) & \label{eq:const-a-c-1}                               \\
     & = \t{const}(c \down{k + k'})    & \explain{by \cref{eq:cp_equal}} \label{eq:const-a-c}
  \end{align}
  The inequality in \cref{eq:const-a-c-1} holds because, when $k' \geq K(a \up k) - k$, the two sides are equal by \cref{eq:ap_equal_2}; when $k' < K(a \up k) - k \leq K(a)$, the left side is always $1$ by the definition of $K(a)$.

  Hence by letting $k' = K(a)$ in \cref{eq:const-a-c}, we have \begin{align}
    K(a) \geq K(c \down k) \geq K(c \down1) - (k-1) \label{eq:ka_kc1}
  \end{align}
  And also \begin{align}
    K(a) \geq K(c \down k)
     & \geq K(d \down k) & \explain{by \cref{eq:kck_geq_kdk}} \label{eq:ka_kd}
  \end{align}

  Now for any $k' \geq K(a)$,
  \begin{align}
    \t{obliv}(a \down{k'})
     & = \t{obliv}(p \down{k + k'}) & \explain{by \cref{eq:ap_equal_2,eq:ka_geq}} \\
     & = \t{obliv}(c \down{k + k'}) & \explain{by \cref{eq:cp_equal}}             \\
     & = \t{obliv}(d \down{k + k'}) & \explain{by \cref{eq:cd_equal_2,eq:ka_kc1}} \\
     & = \t{obliv}(b \down{k'})     & \explain{by \cref{eq:bd_equal,eq:ka_kd}}
  \end{align}

  Similarly for any $k' \geq K(a)$, \begin{align}
    \t{const}(a \down{k'}) = \t{const}(c \down{k'})
  \end{align}

  For any $k' \geq 0$, \begin{align}
    \t{obliv}(a \down{k'})
     & \geq \t{obliv}(p \down{k + k'}) & \explain{by \cref{eq:ap_ob}}    \\
     & = \t{obliv}(c \down{k + k'})    & \explain{by \cref{eq:cp_equal}} \\
     & \geq \t{obliv}(d \down{k + k'}) & \explain{by \cref{eq:cd_ob}}    \\
     & \geq \t{obliv}(b \down{k'})     & \explain{by \cref{eq:bd_ob}}
  \end{align}

  Therefore $a$ and $b$ satisfy the relations in the lemma, leading to a contradiction.
\end{proof}

\begin{corollary}\label{cor:non-const}
  If $a \alias{t} b$ and $b$ is not constant, then $\t{obliv}(a) \geq \t{obliv}(b)$.
\end{corollary}
\begin{proof}
  Since $a \alias{t} b$, assume $a$ and $b$ are aliasing with a common primitive $p$. Since $b$ is not constant, $K(b) = 0$. According to \cref{lem:invariant}, $\t{obliv}(b) = \t{obliv}(p)$ and $\t{obliv}(a) \geq \t{obliv}(p)$. Hence $\t{obliv}(a) \geq \t{obliv}(b)$.
\end{proof}

Now we can establish \cref{thm:dataflow}. Assume the contrary that there is an assignment $a \assigned{t} b$ where $b$ is an oblivious expression, and $a$ is aliasing with a non-oblivious expression $a'$. According to \cref{eq:rule:casting:2} in \cref{rule:casting}, $a$ must be oblivious because it is assigned with an oblivious expression $b$. Now that $a$ is assigned, it cannot be constant. According to \cref{cor:non-const}, $\t{obliv}(a') \geq \t{obliv}(a)$, hence $a'$ is also oblivious, leading to a contradiction. The proof of \cref{thm:dataflow} is complete.

\fi

%% file: chapters/40-implementation.tex
\section{Compiler Optimization Suppression}
\label{sec:suppress}

Modern compilers perform aggressive performance optimizations that reorganize program structures. While \oclang ensures obliviousness at the AST level, the LLVM optimizer may apply transformations that undermine these guarantees.

Consider this example from \cite{sprenkelsLLVMProvidesNo2019}:
\begin{minted}{cpp}
    // Select an element from an array in constant time
    uint64_t constant_time_lookup(
            const size_t secret_idx, const uint64_t table[16]) {
        uint64_t result = 0;
        for (size_t i = 0; i < 16; i++) {
            const bool cond = i == secret_idx;
            const uint64_t mask = (-(int64_t)cond);
            result |= table[i] & mask;
        }
        return result;
    }
\end{minted}

This function uses bit operations to extract \t{table[secret\_idx]}, without revealing \t{secret\_idx} through side channels. However, LLVM's optimizer can detect this pattern and bypass the protection. When compiled with \clang 21.1.0 targeting x86-64 Linux using \t{-O3 -fno-unroll-loops -fno-vectorize}, the compiler generates a branch depending on whether \t{i} equals to \t{secret\_idx}, allowing attackers to learn information about \t{secret\_idx} by observing control flow, violating the intended obliviousness.

This issue arises because current compiler optimizers prioritize correctness without considering obliviousness. Reconstructing LLVM to be obliviousness-aware would require enormous engineering effort.
Instead, we adopt a simpler approach. When our plugin loads, it disables optimization passes that could produce unsafe changes. We manually inspected all \clang optimization passes to determine their safety. Here is the list: ReassociatePass, LoopInstSimplifyPass, LoopSimplifyCFGPass, LoopRotatePass, IndVarSimplifyPass, LoopDeletionPass, LoopInterchangePass, LoopFlattenPass, X86CmovConverterPass.

\ifcr
\else
However, this approach is imperfect. On one hand, by disabling compiler optimization passes, there may be negative performance impact in the resultant code. On the other hand, manual inspection may not discover all relevant passes that need to be disabled, potentially leaving security vulnerabilities.

For the performance impact, we evaluate the impact in \cref{sec:eval:optim-perf} using representative programs. We find there will be only minor slowdown for most programs, and moderate performance degradation (up to 30\%) for complex cases. We regard such slowdown as unavoidable cost to achieve security, and believe the small overheads are acceptable in reality.

For stronger security guarantees, we provide an extra double-check step in our tool. With this step enabled, \oclang applies rigorous obliviousness verification to the generated assembly code using existing tools. Our current implementation uses pitchfork~\cite{disselkoenFindingEliminatingTiming}.
If the verification fails, the compiler falls back to recompile the translation unit with the failed functions, using the \t{optnone} (no optimization) option to avoid any problematic compiler optimizations.
\fi

%% file: chapters/41-evaluation.tex
\section{Evaluation}
\label{sec:eval}

\ifcr
\else
In this section, we evaluate \oclang from two aspects in order to show that \oclang (1) is sufficiently expressive and flexible to support non-trivial oblivious algorithm implementations; (2) has minor overheads during compilation and leads to good performance for the compiled programs.
\fi

\subsection{Case Studies}\label{sec:eval:expr}

\ifcr
\else

To demonstrate the expressiveness of \oclang, we have implemented several oblivious algorithms in C++ and used \oclang to process them. All the implementations can successfully pass the check of \oclang.

\textbf{PathORAM}\label{sec:eval:expr:pathoram}
ORAM is a type of protocol that provides an array-like random access interface but keeps all addresses, data, and operations (read or write) oblivious. It is a useful and general primitive since random array access is a frequently used programming pattern. PathORAM~\cite{stefanov2018path} is a simple and efficient ORAM implementation with $O(\log^3 N)$ time overhead for each access, where $N$ is the array length.

To show \oclang's expressiveness, we write a non-recursive local version of PathORAM, 251 lines of code (LOCs). The implementation is wrapped in a template class with the following public interface:
\begin{minted}{cpp}
    #define OBLIV [[obliv]]
    #define OBLIV_PTR [[obliv("10")]]

    template<typename DataType>
    class PathORAM {
    public:
        PathORAM(size_t N, size_t stashLen);
        ~PathORAM();
        DataType Access(OBLIV AddrType addr, OBLIV DataType data, OBLIV bool isWrite);
    }
\end{minted}

In writing oblivious algorithms, \t{cmove} is a key utility that obliviously copies bytes based on a condition. We implement it via bit operations for cross-platform compatibility.
\begin{minted}{cpp}
    template <typename T>
    void cmove(OBLIV bool cond, OBLIV_PTR const T *src, OBLIV_PTR T *dst);
\end{minted}

The tricky part is handling PathORAM's position map structure. Recall that PathORAM stores data in a binary tree, where each data element is assigned with an position corresponding to a path in the binary tree. On each access to PathORAM, the position corresponding to the accessed address is read out, and then re-assigned to a freshly sampled random value.
The map between data logical addresses and tree path positions, namely the position map, is either stored in an array and accessed with brute-force scanning, or stored in another smaller PathORAM recursively. The content of the position map should be annotated as secret since an attacker aware of the position map can infer the address being accessed from the memory trace. However, when a position is read out from the position map, it should become a public value since it is used to access memory later. To bridge the gap, we should wrap the subroutine of accessing the position map with an unsafe block as below. The security here is guaranteed by the PathORAM theory beyond simple oblivious accesses.
\begin{minted}{cpp}
template<typename DataType>
AddrType PathORAM<DataType>::accessPosMap(OBLIV size_t addr, OBLIV size_t data) {
    OBLIV_UNSAFE { /* return the position and insert random value */ }
}
\end{minted}

In our current implementation, there are some other uses of unsafe blocks, but they could be optionally removed. One usage is the assertion of failure. PathORAM is proved to fail with a negligible probability under properly chosen parameters. An assertion can be used to abort the program when it fails. However, since the failure would expose secret information, the assertion code should be wrapped in an unsafe block. Another usage is the call of library functions, such as \t{std::memcpy}, which can be avoided with hand-written counterparts but with inferior performance.

\fi

\textbf{Oblivious BFS}\label{sec:eval:expr:obfs}
Oblivious BFS~\cite{blantonDataobliviousGraphAlgorithms2013} performs BFS traversals such that an attacker cannot learn the graph structure from the execution trace. It shuffles the graph randomly and chooses each next vertex with equal probability. Our implementation uses 158 LOCs.

\textbf{Real-World Cryptographic Vulnerabilities}\label{sec:eval:expr:cve}
\ifcr
\else
CVE-2024-13176 (the Minerva attack) is a timing side channel in OpenSSL's ECDSA implementation: a $\sim$ \SI{300}{\ns} timing signal leaks when the top word of the inverted nonce is zero, allowing private key recovery on NIST P-521. The vulnerability occurs in \texttt{ec\_field\_inverse\_mod\_ord()}, which computes $k^{-1}$ using \texttt{BN\_mod\_exp\_mont()}.

The root cause was flawed developer reasoning. The original code commented: ``Exponent $e$ is public. No need for constant-time protection.'' While OpenSSL's \texttt{BN\_FLG\_CONSTTIME} flag protects the exponent, it ignores the secret \emph{base} $x$ (the nonce) and secret \emph{result} $r$ (the nonce inverse). The leak came from result normalization via \texttt{bn\_correct\_top()}, which adjusts bignum word count based on the actual value.

With \oclang, this flaw is caught automatically:
\begin{minted}{c}
int BN_mod_exp_mont(BIGNUM *r, const BIGNUM *a, const BIGNUM *p,
                    const BIGNUM *m, BN_CTX *ctx, BN_MONT_CTX *m_ctx);
static int ec_field_inverse_mod_ord(const EC_GROUP *group,
        OBLIV_PTR BIGNUM *r, OBLIV_PTR const BIGNUM *x, BN_CTX *ctx) {
    if (!BN_mod_exp_mont(r, x, ...))  // ERROR: OBLIV_PTR to non-OBLIV_PTR
        goto err;
}
\end{minted}
\fi
\oclang's explicit dataflow tracking catches secret values flowing into unannotated functions, demonstrating robustness over naming conventions alone.

\textbf{Others}
We have also tested \oclang on several other algorithms, including oblivious shuffle~\cite{ohrimenkoMelbourneShuffleImproving2014} and tree-based data structures~\cite{wang2014oblivious}. We also write some standard library wrappers to assist in using standard libraries. The adoption of these programs is straightforward. The only requirement is to mark variables and function parameters with appropriate \olevel{}s.

To quantify the migration effort, we measure the LOCs requiring \texttt{[[obliv]]} annotations compared to the total LOCs. For example, we annotated 32 out of 190 LOCs in PathORAM, 10 out of 94 LOCs in Oblivious BFS, and 25 out of 311 LOCs in Poly1305 (see later). This demonstrates that the annotation burden is modest --- typically just 8\% to 17\% of the codebase requires modifications, and these modifications are localized to variable declarations and function signatures involving secret data.

\subsection{Performance}

\ifcr
\else
In the following experiments, we evaluate the performance of \oclang from three aspects. First, we demonstrate that the check in \oclang introduces minor overheads to the compilation time. Second, we justify that the compiler optimization suppression in \oclang does not significantly degrade the performance of the compiled program. Finally, we compare \oclang with prior solutions to oblivious programming and show its superior performance.

All experiments use an Intel Xeon Gold 5218R machine with \SI{256}{\giga B} DDR4 memory, running Ubuntu 22.04. Single-threaded programs are compiled with \clang 13.0.1 at \t{-O3}; with \oclang, optimization passes mentioned in \cref{sec:suppress} are disabled.

\textbf{Compilation Overheads}\label{sec:eval:compiler-perf}
We evaluate compilation overheads on five benchmarks: \textsf{PathORAM} (251 LOCs, \cref{sec:eval:expr:pathoram}), \textsf{Oblivious BFS} (251 LOCs, \cref{sec:eval:expr:obfs}), \textsf{tiny\_jpeg} (1306 LOCs, JPEG encoding)~\cite{gonzalezTinyJPEG2022}, \textsf{monocypher} (3277 LOCs, cryptography library)~\cite{danielj.bernstein.Monocypher}, and \textsf{sqlite3} (250815 LOCs)~\cite{SQLiteHomePage2023}.

All these programs are compiled (but not linked) with and without our \oclang plugin. The compilation time slowdown results are shown in \cref{fig:compile-time}. We take the average time of 10 runs after 3 warm-up runs.
When not enabling the extra double-check step (\cref{sec:suppress}), the compilation time overheads are less than 11\% for all test cases. Even with the double-check, there are less than 18\% overheads.
Among different test cases, the overheads are smaller for large programs, in which the cost is dwarfed by that of other compilation actions. Therefore, \oclang is particularly suitable for real-world complex oblivious programs.
\fi

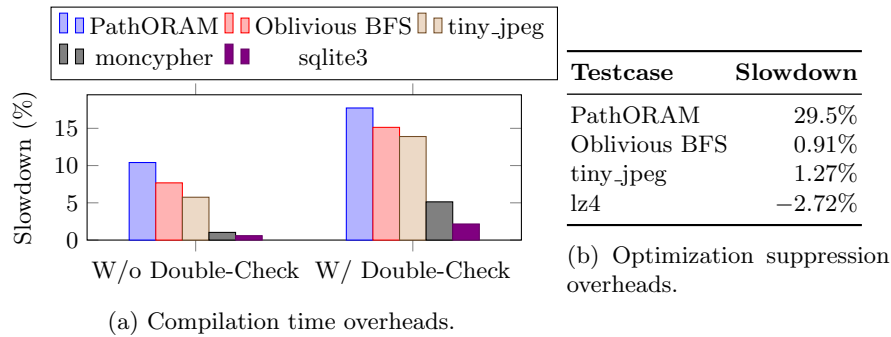
\begin{figure}[htbp]
	\centering
    \small
	\begin{subfigure}{0.6\textwidth}
		\begin{tikzpicture}
			\begin{axis}[
					height = 3.5cm,
					width=\textwidth,
					ybar = 0,
					enlarge x limits = 0.5,
					symbolic x coords = {W/o Double-Check, W/ Double-Check},
					legend columns = 3,
					legend style = {
							at={(0.5, 1.1)},
							anchor = south,
							font=\footnotesize,
						},
					xtick = data,
					ymin = 0,
					scaled y ticks=false,
					ylabel = {Slowdown (\%)},
				]
				\addplot coordinates {
						(W/o Double-Check, 10.41)
						(W/ Double-Check, 17.73)
					};
				\addplot coordinates {
						(W/o Double-Check, 7.68)
						(W/ Double-Check, 15.13)
					};
				\addplot coordinates {
						(W/o Double-Check, 5.75)
						(W/ Double-Check, 13.90)
					};
				\addplot coordinates {
						(W/o Double-Check, 1.03)
						(W/ Double-Check, 5.13)
					};
				\addplot coordinates {
						(W/o Double-Check, 0.59)
						(W/ Double-Check, 2.16)
					};
				\legend{PathORAM, Oblivious BFS, tiny\_jpeg, moncypher, sqlite3};
			\end{axis}
		\end{tikzpicture}
		\caption{Compilation time overheads.}
		\label{fig:compile-time}
		\vspace{-1.5em}
	\end{subfigure}
	\begin{subfigure}{0.35\textwidth}
		\begin{tabular}{lr}
			\toprule
			\bf Testcase  & \bf Slowdown \\
			\midrule
			PathORAM      & $29.5\%$     \\
			Oblivious BFS & $0.91\%$     \\
			tiny\_jpeg    & $1.27\%$     \\
			lz4           & $-2.72\%$    \\
			\bottomrule
		\end{tabular}
		\caption{Optimization suppression overheads.}
		\label{tab:overhead-optim-suppression}
	\end{subfigure}
    \vspace{1.5em}
	\caption{Overheads from compilation process and optimization suppression.}
\end{figure}


\textbf{Optimization Suppression Overheads}\label{sec:eval:optim-perf}
We measure execution time with and without suppression on: \textsf{PathORAM} with randomly generated operations on size $2^{20}$, \textsf{Oblivious BFS} on a 4096-node graph, \textsf{tiny\_jpeg} encoding a $1000\times 1000$ image, and \textsf{lz4}~\cite{LZ4ExtremelyFast} with compression level 3.
\cref{tab:overhead-optim-suppression} shows the results. Most benchmarks see negligible slowdown, with \textsf{lz4} even slightly faster. \textsf{PathORAM} slows down by 29.5\%, likely because the suppressed passes could reorganize its loops for better performance.


\ifcr
\else
\textbf{Runtime Performance Comparison}
Researchers have proposed some other solutions providing oblivious programming utilities. \cref{sec:related} summarizes these designs in great details. Compared with previous solutions, \oclang directly uses the middle-end and backend provided by the mature and heavily optimized \clang and LLVM frameworks, thus achieving superior performance in most cases.

To demonstrate the performance advantages, we consider the following two competitors.
FaCT~\cite{cauligiFaCTDSLTimingsensitive2019b} is a domain-specific language (DSL) that verifies and transforms the obliviousness of source code. It translates the source code into the LLVM IR and relies on \clang to compile to object files.
Jasmin~\cite{almeidaJasminHighAssuranceHighSpeed2017} is also a DSL that ensures the constant-time property. It has a formally verified frontend, optimizer, and codegen in order to translate source code into assembly files.

We evaluate the performance of the above tools and \oclang on six algorithms: \textsf{AES128}, \textsf{poly1305} (one-time authentication), \textsf{salsa20} (stream cipher), \textsf{secretbox} (authenticated encryption), \textsf{siphash} (pseudorandom function~\cite{aumassonSipHashFastShortInput2012}), and \textsf{PathORAM} (block eviction logic).
For the first three cryptogrpahic algorithms, the Jasmin implementation is from libjade~\cite{formosa-cryptoLibjade2023}, a cryptography library providing libsodium-like interfaces written in Jasmin. The implementations in \oclang and FaCT are directly translated from the implementation of libsodium~\cite{Libsodium}, except that the implementation of poly1305 in \oclang is translated from libjade. The implemtations of siphash and PathORAM are all hand-written that take identical control flow and operations. All these algorithms are used to process a 4096-byte message, and we measured the average execution time.

During the implementation, we found that both FaCT and Jasmin are lacking some usability features currently. For example, FaCT has no loop clauses other than ranged \t{for} loops over contingent integers; Jasmin requires manually allocating local variables to registers or the stack. Both of them cannot handle type conversion from boolean values to integers easily, and the compiler error messages are quite vague and hard to locate. These deficiencies make writing complicated programs on them difficult, if not impossible. Therefore, our comparison in this part is limited to a few simple microbenchmarks.

The evaluation results are shown in \cref{fig:perf-compare}.
Among these test cases, \oclang is 27.8\% faster than Jasmin and 66.9\% faster than FaCT on average. For poly1305, the \oclang implementation is 29.8\% slower than Jasmin. The reason might be that the main body of poly1305 calculation is 128-bit integer arithmetic, and the code generation of 128-bit arithmetic in Jasmin is more optimized.

\fi

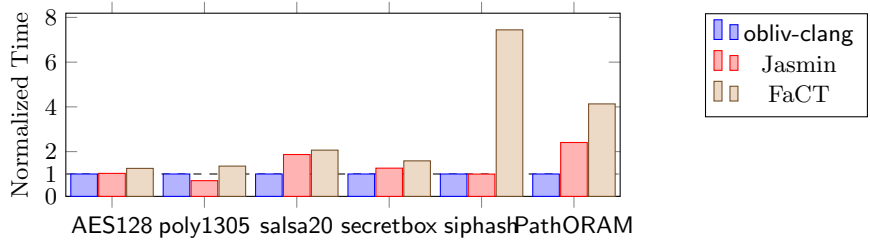
\begin{figure}
  \centering
    \small
    \pgfplotstableread{
      Benchmark  oclang  Jasmin  FaCT
      AES128     1.381   1.419   1.727
      poly1305   1.684   1.182   2.278
      salsa20    3.747   7.009   7.745
      secretbox  6.785   8.555   10.771
      siphash    12.43   12.38   92.52
      PathORAM   51.03   122.88  210.92
    }\datatable

    \begin{tikzpicture}
      \begin{axis}[
          height = 4cm,
          width = 0.73\textwidth,
          ybar = 0.5,
          xtick = data,
          xticklabels from table = {\datatable}{Benchmark},
          xticklabel style={
              font=\sffamily,
            },
          legend style = {
              at = {(1.3, 1)},
              anchor = north,
            },
          scaled y ticks=false,
          ymin = 0,
          ylabel = {Normalized Time},
          enlarge x limits=0.1,
          extra y ticks = {1},
          extra y tick style = {grid = major, grid style={dashed, black}},
        ]
        \addplot table [x expr=\coordindex, y expr=\thisrow{oclang}/\thisrow{oclang}] {\datatable};
        \addplot table [x expr=\coordindex, y expr=\thisrow{Jasmin}/\thisrow{oclang}] {\datatable};
        \addplot table [x expr=\coordindex, y expr=\thisrow{FaCT}/\thisrow{oclang}] {\datatable};

        \legend{\oclang,Jasmin,FaCT};
      \end{axis}
    \end{tikzpicture}
    \caption{Performance comparison of oblivious programs using different tools.}\label{fig:perf-compare}
\end{figure}

%% file: chapters/50-related.tex
\ifcr\else

\section{Related Work}
\label{sec:related}

Researchers have proposed several DSLs that provide compile-time obliviousness guarantees. Jasmin~\cite{almeidaJasminHighAssuranceHighSpeed2017} targets constant-time cryptographic code: programmers annotate secret variables and supply loop invariants, and the compiler uses formal verification tools to check constant-time behavior and invariants; importantly, it offers a fully verified end-to-end compilation flow. ObliVM~\cite{liuObliVMProgrammingFramework2015a} includes a language and compiler that not only checks behaviors but also enforces control-flow obliviousness by transforming programs and inserting dummy operations for secret-dependent branches; it also introduces an \t{rnd} type that restricts use before declassification, but this alone cannot ensure strict security (e.g., it cannot prevent nonce reuse across multiple encryptions). Obliv-c~\cite{zahurOblivCLanguageExtensible2015} extends C by adapting the CIL frontend to translate obliv-c into C, but supports only single-level obliviousness, making it insufficiently expressive for real-world programs, and its rules lack formal security analysis. FaCT~\cite{cauligiFaCTDSLTimingsensitive2019b} generates constant-time code via more aggressive control-flow transformations than ObliVM, automatically rewriting all \t{if} statements and range-based \t{for} loops.

Other prior works focus on checks or transformations on IRs or binaries. ObliCheck~\cite{sonObliCheckEfficientVerification2021} focuses solely on static analysis of memory trace obliviousness, not control-flow obliviousness. Recognizing that previous taint analysis methods were too coarse-grained for accurate checking, ObliCheck employs symbolic execution with sophisticated state merging to accelerate its verification. Raccoon~\cite{raneRaccoonClosingDigital2015} transforms compiled LLVM IR from C source code, obfuscating non-oblivious memory operations using \t{cmove} instructions to hide dependencies on secret information. Constantine~\cite{borrelloConstantineAutomaticSideChannel2021} performas IR-level transformation that lineraizes control flow, with support of JIT transformations and advanced optimizations.

Some works handles a more relaxed model. Liu et al.~\cite{liuMemoryTraceOblivious2013} take a different approach, separating variables into ORAM banks based on access patterns, and leveraging the ORAM properties to achieve oblivious memory accesses. EncLang~\cite{sinhaCompilerVerifierPage2017} addresses page-table attacks by making page access patterns oblivious through data layout and memory access rearrangement during machine code generation.

Our \oclang makes several contributions compared to previous designs:

First, we develop semantics for oblivious programming in a modern industry-standard language, supporting important features that previous oblivious DSLs failed to address, including nested references/pointers, compound types, function parameter passing, and templates.

Second, we propose a workflow requiring minimal modification to existing codebases and toolchains while reliably creating oblivious programs. Programmers need only add necessary annotations and compiler options. The workflow flexibly accommodates structures that cannot be automatically checked, with annotations clearly indicating which program parts require manual auditing.

Third, our evaluation shows that \oclang adds negligible compilation overhead, and our optimization suppression approach has minimal performance impact. Leveraging LLVM and \clang's mature ecosystem, our implementations of oblivious algorithms achieve better performance compared to previous tools.

\fi

%% file: chapters/51-conclusion.tex
\section{Conclusion}\label{sec:conclusion}

We present \oclang, a compiler plugin for the C++ language that performs static checks for oblivious programming. It is specifically implemented using the \clang and LLVM infrastructures.
Compared to previous work, \oclang aims to provide natural semantics and maximum compatibility with widely used industry-level programming languages and workflow while being able to provide a verifiable soundness guarantee even in the presence of complex language features. Our experiments show that \oclang leads to low compilation overheads and can achieve better performance on the compiled oblivious programs compared to previous tools.